\documentclass[12pt,reqno]{amsart}
\usepackage{amsmath,amssymb}
\usepackage{tikz}
\newtheorem{thm}{Theorem}

\newtheorem{lem}[thm]{Lemma}
\theoremstyle{definition}

\title{Perturbation theory of KMS states}
\author{S. Ejima and Y. Ogata}

\begin{document}
\maketitle
\begin{abstract}
We extend the new perturbation formula of equilibrium states by Hastings to KMS states of general $W^*$-dynamical systems.
\end{abstract}
\section{Introduction}

Perturbation theory of equilibrium states is a basic problem in quantum statistical physics which has been studied from old time.
For a finite system with Hamiltonian $H$, equilibrium state with inverse temperature $\beta$ is defined 
by the Gibbs state:
\begin{equation}
    \frac{\mathrm{Tr}(e^{-\beta H}A) }{\mathrm{Tr}(e^{-\beta H})}.
\end{equation}
If we perturb $H$ by $V$,  $e^{-\beta H}$ is replaced by $e^{-\beta(H+V)}$.
Using the Duhamel formula, $e^{-\beta(H+V)}$ can be obtained from $e^{-\beta H}$
by multiplying
\begin{equation}\label{eqtz}
        E_V^\tau\left(\frac{i\beta} 2\right) = \sum_{n\geq 0}\left(-\frac\beta 2\right)^n\int_{0\leq s_n\leq \cdots \leq s_1\leq 1}
        ds_1 \cdots ds_n\ \tau_{\frac{i\beta}2s_n }(V)\cdots \tau_{\frac{i\beta}2 s_1 }(V)
    \end{equation}
    and its adjoint from left and right respectively. 
    Here $\tau$ is the dynamics given by $H$, i.e.,  $\tau_t(A):=e^{itH}Ae^{-itH}$.
    Namely, we have
    \begin{align}\label{duhamel}
    e^{-\beta (H+V)}=E_V^\tau\left(\frac{i\beta} 2\right) e^{-\beta H} 
    \left( E_V^\tau\left(\frac{i\beta} 2\right)\right)^*.
    \end{align}
    In the thermodynamic limit, this $E_V^\tau\left(\frac{i\beta} 2\right)$ can diverge. 
  In spite of that, an analogous representation exists in infinite systems as well, thanks to Araki \cite{A}. (See Theorem \ref{p9}.)
Corresponding to the possible divergence of  $E_V^\tau\left(\frac{i\beta} 2\right)$, the representation by Araki is written in terms of an unbounded operator.

In \cite{H}, Hastings introduced a new representation of $e^{-\beta (H+V)}$
with $E_V^\tau\left(\frac{i\beta} 2\right) $ replaced by some bounded operator $\mathcal O(V)$:
 \begin{align}\label{hastings}
    e^{-\beta (H+V)}=\mathcal O(V)e^{-\beta H} 
    \left( \mathcal O(V)\right)^*.
    \end{align}

The main result of this paper is the extension of Hasting's result to $W^*$-dynamical systems. For the notations and known facts about $W^*$-dynamical systems used in this paper, see Appendix \ref{dyn} and Appendix \ref{kms}.
Let $\mathcal{H}$ be a Hilbert space, $\mathfrak{M} \subset B(\mathcal{H})$ a von Neumann algebra.
We denote the set of self-adjoint elements in $\mathfrak M$ by $\mathfrak M_{\rm sa}$.
The predual of $\mathfrak M$ is denoted by $\mathfrak M_*$.
Let $\Omega\in \mathcal{H}$ be a cyclic and separating unit vector for $\mathfrak{M}$. 
    Define the functional 
    $\omega \in \mathfrak{M}_*$ by  
    $\omega(A) = (\Omega,\ A\Omega)$ for all $A\in \mathfrak{M}$.
Let $\Delta$ be the modular operator and $J$ the modular conjugation associated to $(\mathfrak{M},\ \Omega)$.
Let $\sigma$ be the modular automorphism group, and 
$\sigma^Q$ its perturbation by $Q\in{\mathfrak M}_{\mathrm sa}$.
Define Liouvillean  $L$ of $\sigma$ by $L=\log \Delta$.
(See Appendix \ref{kms}.)
For $V\in \mathfrak{M}_{sa}$, define $\Omega_V = e^{\frac{L+V}{2}}\Omega$ and 
    the functional $\omega_V \in \mathfrak{M}_*$ 
    by $\omega_V(A) = \left(\Omega_V,\ A\Omega_V\right)$ for all $A\in \mathfrak{M}$.
In \cite{A}, it was proven that
$\frac{\omega_V}{\omega_V(1)}$ is a $(\sigma^V,-1)$-KMS state.
The explicit expansion of $\omega_V$ is given, in terms of unbounded operator,  
the Liouvillean $L$. This corresponds to the expansion (\ref{duhamel}).
(See Appendix Theorem \ref{p9} (e).)

Our main Theorem is an alternative representation of $\omega_V$, in terms of bounded operators, which corresponds to the expression by Hastings (\ref{hastings}).

Define the real function $f$ on $\mathbb{R}$ by
    \begin{equation}\label{fdef}
        f(t) = \begin{cases}\displaystyle
            \sum_{n=0}^\infty \frac{2}{n+\frac{1}{2}} e^{-2\pi \left(n+\frac{1}{2}\right)\left|t \right|}
            =-2\log\left(\tanh \frac{\pi|t|}2\right)
            ,
            & (t  \neq 0)\\
            0 & (t = 0)
        \end{cases}
    \end{equation}
    Then, $f\in L^1(\mathbb{R})$.
    Throughout this paper we fix this $f$. Properties of $f$ are collected in Appendix \ref{fprop}.
   
    Define the maps $\Phi: \mathfrak{M}_{sa}\times \mathbb{R} \to \mathfrak{M}_{sa}$  by  
    \begin{equation}
        \Phi(V;u) = \int_\mathbb{R} f\left(t\right)\sigma_t^{uV}\left(V\right) dt,  
    \end{equation}
    for all $V\in \mathfrak{M}_{\mathrm sa}$ and $u\in \mathbb{R}$.
    The right hand side integral is defined in the $\sigma$-weak topology \cite{BR1}.
    Furthermore, $\Phi(V;u)$ is continuous in the norm topology with respect to $u\in {\mathbb R}$.
    (See Lemma \ref{lem8.5}.) Properties of  $\Phi(V;u)$ we use in this paper are shown in Section \ref{pvu}.
    Define 
    $\theta: \mathfrak{M}_{sa} \to \mathfrak{M}$ by 
    \begin{equation}
        \theta(V) = \sum_{n=0}^\infty \int_0^1du_1\int_0^{u_1}du_2\dots\int_0^{u_{n-1}}du_{n}
        \Phi\left(V;u_1\right)\dots\Phi\left(V;u_n\right),
    \end{equation}
    for all $V\in \mathfrak{M}_{\mathrm sa}$. The integral and summation is in the norm topology.
    Here is our main Theorem.
    
\begin{thm}\label{mainthm} 
For any $V\in{\mathfrak M}_{\mathrm sa}$, we have
    \begin{equation}
        \omega_V(A) = (\Omega,\ \theta(V)^*A\theta(V)\Omega),\quad A\in{\mathfrak M}.
    \end{equation}
\end{thm}

\section{Proof of Theorem \ref{mainthm} for analytic $V$}\label{ana}
In this section, we show  Theorem \ref{mainthm} for analytic $V$.We will use facts in Appendix \ref{kms}.
%Throughout this section, we fix 
%$V \in \mathfrak{M}_{sa}\cap \mathfrak{M}_{\sigma}$.
For $s \in \mathbb{R}$, define the vector $\Omega_s := \Omega_{sV} = e^{\frac{L+sV}{2}}\Omega$ and 
the functional $\omega_s \in \mathfrak{M}_*$ by $\omega_s(A) := \omega_{sV}(A) = (\Omega_s,\ A\Omega_s)$ 
for all $A\in \mathfrak{M}$. Especially, we have $\Omega_0 = \Omega$ and $\omega_0 = \omega$.
For $s\in \mathbb{R}$, the $W^*$-dynamics $\sigma^{sV}$ satisfies
$\sigma_t^{sV}(A):=e^{it(L+sV)}Ae^{-it(L+sV)}$ for all $A\in \mathfrak{M}$ and $t\in \mathbb R$.
Define $\Delta_{\Omega_s} = e^{(L+sV-sJVJ)}$ which is the modular operator of the pair 
$(\mathfrak{M},\Omega_s)$. 
By Theorem~\ref{p9}(d), $\Omega_s$ is continuous in the norm topology with respect to $s\in \mathbb{R}$.

\begin{lem}\label{lem1}
    Let $V \in \mathfrak{M}$ be a $\sigma$-entire analytic and self-adjoint element.
     Then 
    $\Omega_s$ is differentiable with respect to $s\in \mathbb{R}$ in the norm topology, and we have 
    \begin{equation}\label{lem2eq}
        \frac{d}{ds} \Omega_s 
    = \frac{1}{2}\int_0^1 du\ \sigma_{-\frac{iu}{2}}^{sV}(V) \Omega_s.
    \end{equation}
\end{lem}
\begin{proof}
%    First, we show that $\Omega \in D\left(e^{u\left(L+sV\right)}\right)$ and that
%    the function $\mathbb{C} \rightarrow \mathcal{H}$ defined by
%    $z \mapsto e^{z\left(L+sV\right)}\Omega$ is entire analytic and 
%    \begin{equation}
%        e^{z\left(L+sV\right)}\Omega = E_{sV} (-iz)\Omega
%    \end{equation}
%    for all $z\in \mathbb{C}$.
%    Consider the following two complex functions on $\mathbb{C}$:
%    \begin{equation}
%        z \mapsto \left(e^{\bar{z}(L+sV)}\xi,\ \Omega\right)    
%    \end{equation} and 
%    \begin{equation}
%        z\mapsto \left(\xi,\ E_{sV}(-iz)\Omega\right).
%    \end{equation}
%    They are analytic on $\mathbb{C}$ and coincide on $i\mathbb{R}$ 
%    since for all $t \in \mathbb{R}$, we have
%    \begin{align}
%        \left(e^{\bar{it}(L+sV)}\xi,\ \Omega\right)
%        &= \left(\xi,\ e^{it(L+sV)}\Omega\right)\nonumber\\
%        &= \left(\xi,\ e^{it(L+sV)}e^{-itL}\Omega\right)\nonumber\\
%        &= \left(\xi,\ E_{sV}(t)\Omega\right).
%    \end{align}
%    The second equality follows by Theorem~\ref{p5}(c) and 
%    the third equality follows by Theorem~\ref{p7}.
%    Therefore they coincide on $\mathbb{C}$.
%    Since $D_s$ is a core for $e^{z(L+sV)}$, 
First we recall  several equalities which can be shown using identity theorem:
We have $\Omega\in D\left(e^{z(L+sV)}\right)$ and
    \begin{equation}\label{hajime}
        e^{z\left(L+sV\right)}\Omega = E_{sV} (-iz)\Omega,
    \end{equation}
    for all $z\in{\mathbb C}$ and $s\in{\mathbb R}$.
    %D\left(e^{z\frac{L+sV}{2}}Ve^{\left(1-z\right)\frac{L+\tilde{s}V}{2}}V\right)\end{equation}
    %, that the function $\mathbb{C} \rightarrow \mathcal{H}$ defined by 
    %$z \mapsto e^{z\left(\frac{L+sV}{2}\right)}e^{\left(1-z\right)\left(\frac{L+\tilde{s}V}{2}\right)}\Omega$ is differentiable,
    %and tha
We also have $\Omega \in D\left(e^{z\frac{L+sV}{2}}e^{\left(1-z\right)\frac{L+\tilde{s}V}{2}}\right)$ for all $z\in \mathbb{C}$ and $s$, $\tilde{s} \in \mathbb{R}$, 
    \begin{equation}\label{tt}
        e^{z\left(\frac{L+sV}{2}\right)}e^{\left(1-z\right)\left(\frac{L+\tilde{s}V}{2}\right)}\Omega
        = E_{sV}\left(-\frac{iz}{2}\right)E_{\tilde{s}V}\left(\frac{i\bar{z}}{2}\right)^* e^{\frac{L+\tilde{s}V}{2}} \Omega.
    \end{equation}
The proofs are analogous to that of (\ref{iib}), (\ref{iia}) below that we omit them.
   
%    Since $D_s$ is a core for $e^{\bar{z}\frac{L+sV}{2}}$, we have 
%    $\Omega \in D\left(e^{z\frac{L+sV}{2}}e^{\left(1-z\right)\frac{L+\tilde{s}V}{2}}\right)$ and 
%    \begin{equation}
%        e^{z\left(\frac{L+sV}{2}\right)}e^{\left(1-z\right)\left(\frac{L+\tilde{s}V}{2}\right)}\Omega
%        = E_{sV}\left(-\frac{iz}{2}\right)E_{\tilde{s}V}\left(\frac{i\bar{z}}{2}\right)^* e^{\frac{L+\tilde{s}V}{2}} \Omega
%    .\end{equation}
    Since the right-hand side of (\ref{tt}) is analytic in norm, the left-hand side is analytic too. Therefore, it is differentiable with respect to $z$ in norm.
    We claim that $\Omega\in D\left(e^{z\frac{L+sV}{2}}Ve^{\left(1-z\right)\frac{L+\tilde{s}V}{2}}\right)$ and
    the derivative is given by 
    \begin{equation}\label{ddee}
        \frac{d}{dz}\left(e^{z\left(\frac{L+sV}{2}\right)}e^{\left(1-z\right)\left(\frac{L+\tilde{s}V}{2}\right)}\Omega\right) = 
        \frac{s-\tilde{s}}{2}e^{z\left(\frac{L+sV}{2}\right)}Ve^{\left(1-z\right)\left(\frac{L+\tilde{s}V}{2}\right)}\Omega
    \end{equation}
    for all $z\in \mathbb{C}$ and $s$, $\tilde{s} \in \mathbb{R}$.
    
     To prove this, define the subspace $D_s \subset \mathcal{H}$ for $s\in\mathbb{R}$, 
     by
    \begin{equation}
        D_s := \bigcup_{N\in\mathbb{N}} E_{[-N,N]}\mathcal{H},
    \end{equation}
    where $\{E_\lambda\}_{\lambda\in\mathbb{R}}$ is the projection-valued measure such that 
    $\int_\mathbb{R} dE_\lambda \lambda = L+sV$.
    Then, $D_s$ is a core for $e^{z(L+sV)}$ for all $z\in \mathbb{C}$.\\
    
    For $\xi \in D_s$, we have
    \begin{align}\label{dde}
        & \left(\xi,
      \frac{d}{dz} e^{{z}\left(\frac{L+sV}{2}\right)} e^{(1-z)\left(\frac{L+\tilde{s}V}{2}\right)}\Omega\right)
      = \frac{d}{dz}\left(e^{\bar{z}\left(\frac{L+sV}{2}\right)}\xi,
        \ e^{(1-z)\left(\frac{L+\tilde{s}V}{2}\right)}\Omega\right)\nonumber
        \\ &= \left(\left(\frac{L+sV}{2}\right)e^{\bar{z}\left(\frac{L+sV}{2}\right)}\xi,\ e^{(1-z)\left(\frac{L+\tilde{s}V}{2}\right)}\Omega\right)
            + \left(e^{\bar{z}\left(\frac{L+sV}{2}\right)}\xi,
            \ -\left(\frac{L+\tilde{s}V}{2}\right)e^{(1-z)\left(\frac{L+\tilde{s}V}{2}\right)}\Omega\right)\nonumber \\
        &= \frac{s-\tilde{s}}{2}\left(e^{\bar{z}\left(\frac{L+sV}{2}\right)}\xi,\ 
        Ve^{(1-z)\left(\frac{L+\tilde{s}V}{2}\right)}\Omega\right)
    .\end{align}
We claim  
\begin{equation}\label{iib}Ve^{(1-z)\left(\frac{L+\tilde{s}V}{2}\right)}\Omega\in 
    D\left(e^{z\left(\frac{L+sV}{2}\right)}\right)\end{equation}
    and
    \begin{align}\label{iia}
        & e^{z\left(\frac{L+sV}{2}\right)}Ve^{(1-z)\left(\frac{L+\tilde{s}V}{2}\right)}\Omega\nonumber\\
        &= E_{sV}^\sigma\left(-\frac{iz}{2}\right)\sigma_{-\frac{iz}{2}}(V)
        \left(E_{sV}^\sigma\left(-\frac{\overline{iz}}{2}\right)\right)^*E_{sV}^\sigma\left(-\frac{iz}{2}\right)
        \left(E_{\tilde{s}V}^\sigma\left(-\frac{\overline{iz}}{2}\right)\right)^*e^{\left(\frac{L+\tilde{s}V}{2}\right)}\Omega
    .\end{align}
    From (\ref{dde}) and (\ref{iib}), we obtain (\ref{ddee}).
    
    To prove (\ref{iib}) and (\ref{iia}), let $\xi$ be an arbitrary element of $D_s$ and
    consider the following two complex functions on $\mathbb{C}$ :
    \begin{equation}\label{iia2}
        z \mapsto \left(e^{\bar{z}\left(\frac{L+sV}{2}\right)}\xi,
        \ Ve^{(1-z)\left(\frac{L+\tilde{s}V}{2}\right)}\Omega\right)
    \end{equation} and \begin{equation}
        z \mapsto \left(\xi,\ E_{sV}^\sigma\left(-\frac{iz}{2}\right)\sigma_{-\frac{iz}{2}}(V)
        \left(E_{sV}^\sigma\left(-\frac{\overline{iz}}{2}\right)\right)^*E_{sV}^\sigma\left(-\frac{iz}{2}\right)
        \left(E_{\tilde{s}V}^\sigma\left(-\frac{\overline{iz}}{2}\right)\right)^*e^{\left(\frac{L+\tilde{s}V}{2}\right)}\Omega\right)
    .\end{equation}
    These two functions coinide on $i\mathbb{R}$ since for $t\in \mathbb{R}$ we have
    \begin{align}
        & \left(e^{\overline{it}\left(\frac{L+sV}{2}\right)}\xi,
        \ Ve^{\left(1-it\right)\left(\frac{L+\tilde{s}V}{2}\right)}\Omega\right) \nonumber\\
        &= \left(\xi,\ e^{it\left(\frac{L+sV}{2}\right)}Ve^{-it\left(\frac{L+\tilde{s}V}{2}\right)}
        e^{\frac{L+\tilde{s}V}{2}}\Omega\right) \nonumber\\
        &= \left(\xi,\ \sigma^{sV}_{\frac{t}{2}}(V)e^{it\left(\frac{L+sV}{2}\right)}e^{-it\frac L2}
        \left(e^{it\left(\frac{L+\tilde{s}V}{2}\right)}e^{-it\frac L2}\right)^*e^{\frac{L+\tilde{s}V}{2}}\Omega\right)\nonumber\\
        &= \left(\xi,\ E_{sV}^{\sigma}\left(\frac{t}{2}\right)\sigma_{\frac{t}{2}}(V)
        \left(E_{sV}^{\sigma}\left(\frac{t}{2}\right)\right)^*E_{sV}^{\sigma}\left(\frac{t}{2}\right)
        \left(E_{\tilde{s}V}^{\sigma}\left(\frac{t}{2}\right)\right)^*e^{\frac{L+\tilde{s}V}{2}}\Omega\right).
    \end{align}
    The second equality follows by (\ref{tlq}) and 
    the third equality follows due to (\ref{ete}), (\ref{elq}).
Since these two functions are analytic (recall (\ref{eqexp}) and (\ref{hajime}) ), they coincide on $\mathbb{C}$.
    Since $D_s$ is a core for $e^{\bar{z}\left(\frac{L+sV}{2}\right)}$, 
    we obtain the claim (\ref{iib}) and (\ref{iia}).
Hence we have completed the proof of (\ref{ddee}).

   Finally, we show (\ref{lem2eq}).
   To prove this, first note that from (\ref{iia}) and the expansion of $E_{sV}^{\sigma}(z)$ (\ref{eqtz}), we obtain the following continuity in the norm:
\begin{align}\label{conti}
\lim_{s\to \tilde s}
e^{z\left(\frac{L+sV}{2}\right)}Ve^{(1-z)\left(\frac{L+\tilde{s}V}{2}\right)}\Omega=
e^{z\left(\frac{L+\tilde sV}{2}\right)}Ve^{(1-z)\left(\frac{L+\tilde{s}V}{2}\right)}\Omega.
\end{align}
For
$s,\tilde s\in\mathbb{R}$, since the map 
    \begin{equation}
       \mathbb{R}\ni  u \mapsto e^{u\left(\frac{L+sV}{2}\right)}e^{(1-u)\left(\frac{L+\tilde{s}V}{2}\right)}\Omega  \in \mathcal{H}   
    \end{equation}
    is differentiable in the norm topology and satisfies
    \begin{align}
        &\frac{d}{du} e^{u\left(\frac{L+sV}{2}\right)}e^{(1-u)\left(\frac{L+\tilde{s}V}{2}\right)}\Omega
        = \left(\frac{s-\tilde{s}}{2}\right)e^{u\left(\frac{L+sV}{2}\right)}V
        e^{(1-u)\left(\frac{L+\tilde{s}V}{2}\right)}\Omega
    ,\end{align}
(from (\ref{ddee})),
    we have 
    \begin{align}
        \Omega_s -\Omega_{\tilde{s}} 
        &= e^{\frac{L+sV}{2}}\Omega - e^{\frac{L+\tilde{s}V}{2}}\Omega \nonumber\\
        &= \int_0^1du\ \frac{d}{du} 
        e^{u\left(\frac{L+sV}{2}\right)}
        e^{(1-u)\left(\frac{L+\tilde{s}V}{2}\right)}\Omega \nonumber\\
        &= \frac{s-\tilde{s}}{2}
\int_0^1 du\ e^{u\left(\frac{L+sV}{2}\right)}V
        e^{(1-u)\left(\frac{L+\tilde{s}V}{2}\right)}\Omega \nonumber\\
           \end{align}
    Therefore from the continuity (\ref{conti}) and (\ref{iia}), using the  Lebesgue's dominated convergence theorem, $\Omega_s$ is differentiable with respect to $s$ in the norm topology and 
    we have
    \begin{align}
        \frac{d}{ds}\Omega_s
        &= \frac{1}{2}\int_0^1du\ e^{u\left(\frac{L+sV}{2}\right)}Ve^{(1-u)\left(\frac{L+sV}{2}\right)}\Omega\nonumber\\
        &= \frac{1}{2}\int_0^1du\ \sigma_{-\frac{iu}{2}}^{sV}(V)\Omega_s.
    \end{align}

\end{proof}
\begin{lem}\label{lem2}
    Let $V \in \mathfrak{M}$ be a self-adjoint $\sigma$-entire analytic element,  
    and $A\in  \mathfrak{M}$. Then $\omega_s(A)$ is differentiable with respect to $s$ and
    \begin{equation} 
    \frac{d}{ds} \omega_s\left(A\right)= \int_0^1 du\left(\Omega_s, \sigma^{sV}_{iu}\left(V\right)A\Omega_s\right)
\end{equation}
for all $s\in \mathbb{R}$.
    
\end{lem}
\begin{proof}
    By Lemma~\ref{lem1}, $\omega_s(A)$ is differentiable and we have
    \begin{align}
    \frac{d}{ds} \omega_s(A) 
    &= \frac{d}{ds}\left(\Omega_s,\ A\Omega_s\right)
    \nonumber \\
    &= \left(\frac{1}{2}\int_0^1du\ \sigma^{sV}_{-\frac{iu}{2}}(V)\Omega_s,\ A\Omega_s\right) + 
    \left(\Omega_s,\ A\left(\frac{1}{2}\int_0^1du\ \sigma_{-\frac{iu}{2}}^{sV}(V)\Omega_s\right)\right)
    \nonumber \\
    &= \frac{1}{2}\int_0^1du\left(\left(\Omega_s,\ \sigma_{\frac{iu}{2}}^{sV}(V)A\Omega_s\right)+
    \left(\Omega_s,\ A\sigma_{-\frac{iu}{2}}^{sV}(V)\Omega_s\right)  \right).
    \end{align}
    We obtain 
    \begin{equation}
        \left(\Omega_s,\ A\sigma_{-\frac{iu}{2}}^{sV}(V)\Omega_s\right)
        = \left(\Omega_s,\ \sigma_{i\left(1-\frac{u}{2}\right)}^{sV}(V)A\Omega_s\right)
    \end{equation}
    since we have
    \begin{align}
        &\left(\Omega_s,\ A\sigma_{-\frac{iu}{2}}^{sV}(V)\Omega_s\right)
        = \left(A^*\Omega_s,\ \sigma_{-\frac{iu}{2}}^{sV}(V)\Omega_s\right)
        \nonumber \\
        &= \left(S_{\Omega_s}A\Omega_s,\ S_{\Omega_s}\sigma_{\frac{iu}{2}}^{sV}(V)\Omega_s\right)
         = \left(\Delta_{\Omega_s}^{\frac{1}{2}}\sigma_{\frac{iu}{2}}^{sV}(V)\Omega_s,
        \ \Delta_{\Omega_s}^{\frac{1}{2}}A\Omega_s\right)
         \nonumber\\
        &= \left(\Delta_{\Omega_s}^{1-\frac{u}{2}}V\Omega_s,\ A\Omega_s\right) 
         = \left(\sigma_{-i\left(1-\frac{u}{2}\right)}^{sV}(V)\Omega_s,\ A\Omega_s\right)
        \nonumber\\
        &= \left(\Omega_s,\ \sigma_{i\left(1-\frac{u}{2}\right)}^{sV}(V)A\Omega_s\right)
        .
    \end{align}
    Therefore we have
    \begin{align}
        \frac{d}{ds} \omega_s(A) 
        &=  \frac{1}{2}\int_0^1du\left(\Omega_s,\ \sigma_{\frac{iu}{2}}^{sV}(V)A\Omega_s\right)+
        \frac{1}{2}\int_0^1du \left(\Omega_s,\ \sigma_{i\left(1-\frac{u}{2}\right)}^{sV}(V)A\Omega_s\right)
         \nonumber\\
        &= \int_0^{\frac{1}{2}}du\left(\Omega_s,\ \sigma_{iu}^{sV}(V)A\Omega_s\right)
        + \int_{\frac{1}{2}}^{1}du\left(\Omega_s,\ \sigma_{iu}^{sV}(V)A\Omega_s\right)
         \nonumber\\
        &= \int_{0}^{1}du
        \left(\Omega_s,\ \sigma_{iu}^{sV}(V)
        A\Omega_{s}\right).
    \end{align}
\end{proof}
%    By Lemma~\ref{lem4}, $\Delta_{\Omega_s}^u V\Omega_s$ is weakly continuous with respect to $u\in \mathbb{R}$. Hence we can define $\int_0^1 du \Delta_{\Omega_s}^u V\Omega_s \in \mathcal{H}$ by
%    \begin{equation}
%        \left(\int_0^1 du \Delta_{\Omega_s}^u V\Omega_s, \xi\right) = 
%        \int_0^1 du \left(\Delta_{\Omega_s}^u V\Omega_s, \xi\right)
%    \end{equation}
%    for all $\xi \in \mathcal{H}$.\\
%    For $f\in L^1(\mathbb{R})$ and $V\in \mathfrak{M}_{sa}$, we can define 
%    $\displaystyle \int_\mathbb{R} f(t)\sigma_t^{sV}(V)dt \in \mathfrak{M}_{sa}$ by
%    \begin{equation}
%        \eta\left(\int_\mathbb{R} f(t)\sigma_t^{sV}(V)dt\right) = 
%        \int_\mathbb{R} f(t)\eta\left(\sigma_t^{sV}(V)\right)dt
%    \end{equation}
%    for all $\eta \in \mathfrak{M}_*$.
\begin{lem}\label{lem6}
    Let $V \in \mathfrak{M}$ be a $\sigma$-entire analytic and self-adjoint element,
    and $A\in \mathfrak{M}$.
    Let $f\in L^1(\mathbb R)$ be the function defined in (\ref{fdef}). 
    Then for all $s\in \mathbb{R}$, we have
    \begin{equation}\label{lem6eq}
        \int_0^1du \left(\Omega_s, \sigma_{iu}^{sV}\left(V\right) A\Omega_s\right) = 
        \left(\Phi(V;s)\Omega_s, A\Omega_s\right) + \left(\Omega_s, A\Phi(V;s)\Omega_s\right),
    \end{equation} 
    where 
    \begin{equation}
        \Phi(V;s) = \int_\mathbb{R} f\left(t\right)\sigma_t^{sV}\left(V\right) dt    
    .\end{equation}
\end{lem}
\begin{proof}
    Let $G$ be a function defined on $\mathbb{R}_{>0}$ by
    \begin{equation}
        G(\lambda) = (1+\lambda)^{-1}\int_0^1du\ \lambda^u,\quad \lambda>0.
    \end{equation}
    Define a function $F$ on $\mathbb{R}$ by
    \begin{equation}\label{Fdef}
        F(x) = \begin{cases}\displaystyle
            \frac{e^x-1}{e^x+1} \frac{1}{x} & (x \neq 0)   \\
            \\
            \displaystyle\frac{1}{2} & (x = 0)
        \end{cases}.
    \end{equation}
      It follows that $F(x) = G\left(e^x\right)$ for all $x\in \mathbb{R}$.
   Therefore, from (\ref{ff}) in Appendix \ref{fprop}, we have
    \begin{equation}
        G\left(e^x\right) = \int_{\mathbb{R}} f(t)e^{itx}dt,
    \end{equation}
    for all $x\in\mathbb  R$.
%    Define $\Phi(V;s)$ for all $s\in \mathbb{R}$ by
%    \begin{equation}
%        \Phi(V;s) = \int_\mathbb{R} f\left(t\right)\sigma_t^{sV}\left(V\right) dt    .
%    \end{equation}
    We now show (\ref{lem6eq}) for any $\sigma$-entire analytic element $A$.
    We have 
    \begin{align}
        \int_0^1 du\left(\Omega_s,\ \sigma_{iu}^{sV}(V)A\Omega_s\right)
        &= \int_0^1du\left(\sigma_{-iu}^{sV}(V)\Omega_s,\ A\Omega_s\right)
         \nonumber\\
        &= \left(\int_0^1du \Delta_{\Omega_s}^uV\Omega_s,\ A\Omega_s\right)
         \nonumber\\
        &= \left(\int_0^1du\Delta_{\Omega_s}^uV\Omega_s,\ 
        (1+\Delta_{\Omega_s})^{-1}(1+\Delta_{\Omega_s})A\Omega_s\right)
         \nonumber\\
        &= \left((1+\Delta_{\Omega_s})^{-1}\int_0^1du\Delta_{\Omega_s}^uV\Omega_s,\ 
        (1+\Delta_{\Omega_s})A\Omega_s\right).
    \end{align}
    We have $A\Omega_s\in D(\Delta_{\Omega_s})$ due to the $\sigma$-entire analyticity of $A$.
    
    We have
    \begin{align}
        \left(1+\Delta_{\Omega_s}\right)^{-1} \int_0^1 du \Delta_{\Omega_s}^u V\Omega_s
        = G\left(\Delta_{\Omega_s}\right)V\Omega_s
        = \int_{\mathbb{R}} dt f(t)\Delta_{\Omega_s}^{it}V\Delta_{\Omega_s}^{-it}\Omega_s= \Phi(V;s)\Omega_s .
    \end{align}
    Therefore we have
    \begin{align}
        \int_0^1du \left(\Omega_s, \sigma_{iu}^{sV}\left(V\right) A\Omega_s\right)
        &= \left((1+\Delta_{\Omega_s})^{-1}\int_0^1du\Delta_{\Omega_s}^uV\Omega_s,\
        (1+\Delta_{\Omega_s})A\Omega_s\right)\nonumber\\ 
        &= \big(\Phi(V;s)\Omega_s,\ (1+\Delta_{\Omega_s})A\Omega_s\big)\nonumber\\
        &= (\Phi(V;s)\Omega_s,\ A\Omega_s) + (\Phi(V;s)\Omega_s,\ \Delta_{\Omega_s}A\Omega_s))\nonumber\\
        &= (\Omega_s,\ \Phi(V;s)A\Omega_s) + 
        (\Delta^{\frac{1}{2}}_{\Omega_s}\Phi(V;s)\Omega_s,\ \Delta_{\Omega_s}^{\frac{1}{2}}A\Omega_s)\nonumber\\
        &= (\Omega_s,\ \Phi(V;s)A\Omega_s) +
        (J\Delta_{\Omega_s}^{\frac{1}{2}}A\Omega_s,\ J\Delta_{\Omega_s}^{\frac{1}{2}}\Phi(V;s)\Omega_s)\nonumber\\
        &=(\Omega_s,\ \Phi(V;s)A\Omega_s) + 
        (\Omega_s,\ A\Phi(V;s)\Omega_s).
    \end{align}
    We repeatedly used the self-adjointness of $\Phi(V;s)$.
    
    To extend (\ref{lem6eq}) to general $A\in{\mathfrak M}$, we just need to notice that
    both sides of (\ref{lem6eq}) are continuous in $A$ with respect to the strong-topology
    and recall that ${\mathfrak M}_\sigma$ is strong-dense in $\mathfrak M$.
\end{proof}

Due to Lemma \ref{lem8.5}, we can define the map $\theta: \mathfrak{M}_{sa} \times \mathbb{R} \to \mathfrak{M}_{sa}$ by
\begin{equation}
    \theta(V;s) = \sum_{n=0}^\infty \int_0^sdu_1\int_0^{u_1}du_2\dots\int_0^{u_{n-1}}du_{n}
    \Phi\left(V;u_1\right)\dots\Phi\left(V;u_n\right).
\end{equation}
for all $V\in \mathfrak{M}_{sa}$ and $s \in \mathbb{R}$.

We now prove the statement of Theorem \ref{mainthm} for $\sigma$-entire analytic $V$.
\begin{lem}\label{lem7}
    Let $V \in \mathfrak{M}$ be a $\sigma$-entire analytic and self-adjoint element. Then we have
    \begin{equation}
        \omega_s\left(A\right) = \omega\left(\theta(V;s)^* A \theta(V;s)\right),\quad s\in{\mathbb R},\quad
        A\in{\mathfrak M}.
    \end{equation}
\end{lem}
\begin{proof}
For each $s\in{\mathbb R}$, we define the map $\mathcal{L}_{s}: \mathfrak{M} \rightarrow \mathfrak{M}$ by
    \begin{equation}
        \mathcal{L}_{s}(X) = \Phi(V;s)X + X\Phi(V;s),\quad X\in{\mathfrak M}.
    \end{equation}

    By Lemma~\ref{lem2} and Lemma~\ref{lem6}, for any $A\in{\mathfrak M}$, we have 
    \begin{align}
        \frac{d}{ds} \omega_s(A)
        &= \int_0^1du \left(\Omega_s,\ \sigma_{iu}^{sV}(V)A\Omega_s\right)\nonumber\\
        &= \left(\Phi(V;s)\Omega_s,\ A\Omega_s\right) + \left(\Omega_s,\ A\Phi(V;s)\Omega_s\right).\nonumber\\
        &=\omega_s \circ \mathcal{L}_{s}(A).
    \end{align}
          
    Next, define the state $\phi_s\in \mathfrak{M}_*$ by
    \begin{equation}
        \phi_s(A) = \omega(\theta(V;s)^*A\theta(V;s)),\quad A\in{\mathfrak M}.
    \end{equation} 
    We now show that $\phi_s$ satisfies the same differential equation as $\omega_s$:
    \begin{equation}\label{50}
        \frac{d}{ds} \phi_s(A)
         = \phi_s \circ \mathcal{L}_s(A),\quad A\in{\mathfrak M}.
    \end{equation}
    By Lemma~\ref{lem8.5}, the map
    $\mathbb{R} \ni s\mapsto \Phi(V;s)\in \mathfrak{M}$ is continuous in the norm topology. Therefore $\theta(V;s)$ is differentiable 
    in the norm topology and we have 
    \begin{equation}
        \frac{d}{ds} \theta(V;s) = \Phi(V;s) \theta(V;s).
    \end{equation}
    Since $\Phi(V;s)$ is self-adjoint, we obtain (\ref{50}):
    \begin{align}
        \frac{d}{ds}\phi_s(A)
        &= \omega\left(\theta(V;s)^*\Phi(V;s)^* A\theta(V;s)\right) 
        + \omega\left(\theta(V;s)^*A\Phi(V;s)\theta(V;s)\right)\nonumber\\
        &= \omega\left(\theta(V;s)^*\mathcal{L}_s(A)\theta(V;s)\right)\nonumber\\
        &= \phi_s \circ \mathcal{L}_s(A).
    \end{align}
   As $\phi_s$ and $\omega_s$ satisfies the same differential equation and $\omega_0=\phi_0$,
      we obtain $ \omega_s=\phi_s $ for all $s\in \mathbb{R}$.
\end{proof}
\section{Proof of Theorem \ref{main}}\label{main}
In order to extend the result to general $V$, we need the following continuity of
$\theta(V;s)\Omega$.
\begin{lem}\label{lem9}
    Let $\{V_m\}_{m\in \mathbb{N}}  \subset \mathfrak{M}_{sa}$  be a sequence such that
    $V_m \rightarrow V\in \mathfrak{M}_{\mathrm sa}$ strongly as $m\to \infty$. Then we have 
    $\|\theta(V_m; s)\Omega- \theta(V;s)\Omega\| \to 0$. 
\end{lem}
\begin{proof}
 From the uniform boundedness principle, we have $c:=\sup_m\left\Vert V_m\right\Vert<\infty$.
By Lemma \ref{lem8.5},
$\Phi(V_m; s)$ converges to $\Phi(V;s)$ strongly.
From this convergence and the uniform boundedness
\begin{align}
\sup_m\|\Phi(V_m; s)\|\le \| f\|_{L^1}\sup_m\|V_m\|\le c\| f\|_{L^1},
\end{align}
we have
\begin{align}\label{mono}
\lim_{m\to \infty}
 \left\Vert
    \Phi\left(V_m;u_1\right)\dots\Phi\left(V_m;u_n\right)\Omega
- \Phi\left(V;u_1\right)\dots\Phi\left(V;u_n\right)\Omega
\right\Vert
=0.
\end{align}
We also have
\begin{align}\label{hana}
 \left\Vert
    \Phi\left(V_m;u_1\right)\dots\Phi\left(V_m;u_n\right)\Omega
- \Phi\left(V;u_1\right)\dots\Phi\left(V;u_n\right)\Omega
\right\Vert
\le\left(
c^{n}+\| V\|^n
\right)\|f\|_{L^1}^n,
\end{align}
for any $u_1,\ldots, u_n\in \mathbb R$ and $n\in\mathbb N$.
From the definition of $\theta(V;s)$, we have 
\begin{align}
&\left\Vert \theta(V_m; s)\Omega- \theta(V;s)\Omega\right\Vert\\
&\le
\sum_{n=0}^\infty \int_0^sdu_1\int_0^{u_1}du_2\dots\int_0^{u_{n-1}}du_{n}
  \left\Vert\left(
    \Phi\left(V_m;u_1\right)\dots\Phi\left(V_m;u_n\right)
- \Phi\left(V;u_1\right)\dots\Phi\left(V;u_n\right)\right)
\Omega
\right\Vert
\end{align}
From (\ref{mono}) and (\ref{hana}), applying Lebesgue dominated convergence Theorem, we obtain
\begin{align}
\lim_{m\to\infty}\left\Vert \theta(V_m; s)\Omega- \theta(V;s)\Omega\right\Vert=0.
\end{align}
   \end{proof} 
\begin{proof}[Proof of Theorem~\ref{mainthm}]
    Note that $\theta(V) = \theta(V;1)$.
    There exists a sequence $\{V_m\} \subset \mathfrak{M}_{sa}$ such that 
    $V_m$ is $\sigma$-entirely analytic for all $m\in\mathbb{N}$ and that $V_m \to V$ strongly.
    Let $A\in \mathfrak{M}$.
    By Lemma~\ref{lem9}, we have
    \begin{equation}
        \left(\Omega,\ \theta(V_m)^*A\theta(V_m)\Omega\right) \to 
        \left(\Omega,\ \theta(V)^*A\theta(V)\Omega\right)
    \end{equation}
    as $m\to\infty$.\\
    By Theorem~\ref{p9}(d) , $\omega_{V_m}$ converges to $\omega_{V}$ in the norm topology.
    By Lemma~\ref{lem7} we have
    \begin{equation}
        \omega_{V_m}(A) = \left(\Omega,\ \theta(V_m)^*A\theta(V_m)\right).
    \end{equation} 
    for all $m\in \mathbb{N}$.
    Therefore we conclude that
    \begin{equation}
        \omega_V(A) = \left(\Omega,\ \theta(V)^*A\theta(V)\Omega\right).  
    \end{equation}
\end{proof}

\section{Properties of $ \Phi(V;s) $}\label{pvu}
In this section, we collect some properties of $\Phi$ used in Section \ref{ana} and Section\ref{main}.
\begin{lem}\label{lem8.5}
\begin{enumerate}
\item  For all $V\in \mathfrak{M}_{\mathrm sa}$  and $s\in\mathbb R$, we have
\begin{align}
\|  \Phi(V;s)\|
\le\|f\|_{L^1}\|V\|.
\end{align}
\item For all $V\in \mathfrak{M}_{\mathrm sa}$ the map
    defined by $\mathbb{R}\ni s \mapsto \Phi(V;s)\in \mathfrak{M}$ is norm-continuous.
    \item 
    For any $s\in{\mathbb R}$ and sequence $V_N$ in $\mathfrak M_{\mathrm sa}$ such that
    $V_N\to V\in{\mathfrak M}_{sa}$ in the strong operator topology, we have
    $\Phi(V_N;s)\to \Phi(V;s)$ in the strong operator topology.  
\end{enumerate}
\end{lem}
\begin{proof}
   For each $V\in{\mathfrak M}_{\mathrm sa}$, $t\in\mathbb R$ and $n=0,1,\ldots$,  set 
   \begin{align}
        {\mathcal T}_n(V;t):=  i^n \int_{0\leq s_n\leq \cdots \leq s_1\leq t}        ds_1 \cdots ds_n
       \left(
        \sigma_{s_n}(V)\cdots \sigma_{s_1}(V)
        \right).
   \end{align}
       Here the integral is defined in $\sigma$-weak topology.
       Note that 
       %from the strong continuity of 
       %$\mathbb R\ni s\mapsto \sigma_s(V)\in{\mathfrak M}$
       the map $\mathbb R\ni t\mapsto {\mathcal T}_n(V;t)\in{\mathfrak M}$ is norm-continuous.
       Furthermore, its norm is bounded
       \begin{align}\label{tb}
       \left\Vert {\mathcal T}_n(V;t)\right\Vert\le\frac{|t|^n\left\Vert V\right\Vert^n}{n!}.
       \end{align}
       We have
       \begin{align}
        E_{sV}^\sigma (t)=\sum_{n\geq 0} s^n{\mathcal T}_n(V;t),
       \end{align}
       where the summation converges in norm.
       From (\ref{ete}), we have
        \begin{equation}\label{nm}
            \sigma_t^{sV}(A) = E_{sV}^\sigma(t)\sigma_t(A)E_{sV}^\sigma(t)^*
            =\sum_{n\geq 0} \sum_{m\geq 0}s^{n+m}{\mathcal T}_n(V;t)
            \sigma_t(A)
            \left({\mathcal T}_m(V;t)\right)^*,
        \end{equation}
       where the summation in the last term converges in the norm topology.

For each $\varepsilon>0$, we fix $M_\varepsilon>0$ such that
\begin{align}\label{me}
\int_{|t|\ge M_\varepsilon}dt  \left\vert f(t)\right\vert <\varepsilon. 
\end{align}
    
    Now let us prove the claim of the Lemma. The first statement of the Lemma is trivial.
    Let us prove the second statement of the Lemma.
    For all $V\in{\mathfrak M}_{sa}$, $\xi\in{\mathcal H}$, $s,\tilde s\in\mathbb R$
    and $\varepsilon>0$, using (\ref{tb}), (\ref{me}), we have
    \begin{align}
        &\left\Vert \left(\Phi(V;\tilde{s}) - \Phi(V;s)\right)\xi \right\Vert \nonumber\\
        &\le \int_{|t|\le M_\varepsilon}dt\ f(t)
        \left\Vert \left( \sigma_t^{\tilde{s}V}(V)-\sigma_t^{sV}(V)\right)\xi \right\Vert
    +\int_{|t|\ge M_\varepsilon}dt\ 
       f(t)\left\Vert \left(\sigma_t^{\tilde{s}V}(V)-\sigma_t^{sV}(V)\right)\xi\right\Vert\nonumber\\
    &\le 
      \int_{|t|\le M_\varepsilon}
      dt\ f(t)
       \sum_{n\geq 0} \sum_{m\geq 0}\left| \tilde s^{n+m}-s^{n+m}\right|\left\Vert
       \left({\mathcal T}_n(V;t)
            \sigma_t(V)
            \left({\mathcal T}_m(V;t)\right)^*\right)\xi\right\Vert
            +2\varepsilon \left\Vert \xi\right\Vert    \left\Vert V\right\Vert     \nonumber\\     
    &     \le 
   \left( \sum_{n\geq 0} \sum_{m\geq 0}\left\vert \tilde s^{n+m}-s^{n+m}\right\vert
    \left\Vert f\right\Vert_{L^1}\frac{|M_\varepsilon|^{n+m}}{n!m!}\left\Vert V\right\Vert^{n+m+1}
            +2\varepsilon     \left\Vert V\right\Vert  \right)         \left\Vert \xi\right\Vert
    \end{align}
    Therefore we have 
    \begin{equation}
        \|\Phi(V;\tilde{s}) - \Phi(V;s)\| \leq 
   \sum_{n\geq 0} \sum_{m\geq 0}\left\vert \tilde s^{n+m}-s^{n+m}\right\vert
    \left\Vert f\right\Vert_{L^1}\frac{|M_\varepsilon|^{n+m}\left\Vert V\right\Vert^{n+m+1}}{n!m!}
            +2\varepsilon     \left\Vert V\right\Vert,
    \end{equation}
    for any $\varepsilon>0$, $s,\tilde s\in{\mathbb R}$ and $V\in{\mathfrak M}_{\mathrm sa}$. 
    As $\tilde s\to s$, the first term converges to $0$ by Lebesgue dominated convergence theorem and we obtain
    \begin{align}
 \limsup_{\tilde s\to s}   \|\Phi(V;\tilde{s}) - \Phi(V;s)\| \leq 
 2\varepsilon     \left\Vert V\right\Vert,
    \end{align}
   for any $\varepsilon>0$, $s\in{\mathbb R}$ and $V\in{\mathfrak M}_{\mathrm sa}$. Therefore, we have
   \begin{align}
    \lim_{\tilde s\to s}   \|\Phi(V;\tilde{s}) - \Phi(V;s)\| =0,
   \end{align}
   for any $s\in\mathbb R$ and $V\in{\mathfrak M}_{\mathrm sa}$.
   This proves the second statement of Lemma.
   
   To prove the third statement, let  $V_N$ in $\mathfrak M_{\mathrm sa}$ be a sequence such that
    $V_N\to V\in{\mathfrak M}_{sa}$ in the strong operator topology.
   By the uniform boundedness principle, we have $c:=\sup_N\left\Vert V_N\right\Vert<\infty$.
   Fix arbitrary $\varepsilon>0$, $\xi\in \mathcal H$, and $s\in\mathbb R$.
   We have
 \begin{align}\label{fuku}
        &\left\Vert \left(\Phi(V_N;{s}) - \Phi(V;s)\right)\xi \right\Vert \nonumber\\
        &\le \int_{|t|\le M_\varepsilon}dt\ f(t)
        \left\Vert \left( \sigma_t^{{s}V_N}(V_N)-\sigma_t^{sV}(V)\right)\xi \right\Vert
    +\int_{|t|\ge M_\varepsilon}dt\ 
       f(t)\left\Vert \left(\sigma_t^{sV_N}(V_N)-\sigma_t^{sV}(V)\right)\xi\right\Vert\nonumber\\
    &\le 
      \int_{|t|\le M_\varepsilon}
      dt\ f(t)
       \sum_{n\geq 0} \sum_{m\geq 0}\left| s^{n+m}\right|\left\Vert
       \left({\mathcal T}_n(V_N;t)
            \sigma_t(V_N)
            \left({\mathcal T}_m(V_N;t)\right)^*-
            {\mathcal T}_n(V;t)
            \sigma_t(V)
            \left({\mathcal T}_m(V;t)\right)^*\right)\xi\right\Vert \nonumber\\     
            &+\varepsilon \left\Vert \xi\right\Vert    (\left\Vert V\right\Vert  +c) , 
        \end{align}
        for all $N\in\mathbb N$.
Note that
\begin{align}\label{ttt}
&\left\Vert
       \left({\mathcal T}_n(V_N;t)
            \sigma_t(V_N)
            \left({\mathcal T}_m(V_N;t)\right)^*-
            {\mathcal T}_n(V;t)
            \sigma_t(V)
            \left({\mathcal T}_m(V;t)\right)^*\right)\xi\right\Vert\nonumber\\
&       \le
       \int_{0\leq s_n\leq \cdots \leq s_1\leq t}        ds_1 \cdots ds_n
       \int_{0\leq u_n\leq \cdots \leq u_1\leq t}        du_1 \cdots du_m
      \nonumber\\
     & \left\Vert\left(
       \sigma_{s_n}(V_N)\cdots \sigma_{s_1}(V_N)  \sigma_t(V_N) \sigma_{u_1}(V_N)\cdots \sigma_{u_m}(V_N)-
        \sigma_{s_n}(V)\cdots \sigma_{s_1}(V)  \sigma_t(V) \sigma_{u_1}(V)\cdots \sigma_{u_n}(V)
        \right)\xi
        \right\Vert.
\end{align}
From Lemma \ref{p10}, the integrand on the last line converges to $0$ as $N\to\infty$, point wise.
On the other hand,
they are bounded by $\Vert V\Vert^{n+m+1}+c^{n+m+1}$.
Therefore, from Lebesgue dominated convergence Theorem, we see that 
\begin{align}\label{aza}
\lim_{N\to\infty}
\left\Vert
       \left({\mathcal T}_n(V_N;t)
            \sigma_t(V_N)
            \left({\mathcal T}_m(V_N;t)\right)^*-
            {\mathcal T}_n(V;t)
            \sigma_t(V)
            \left({\mathcal T}_m(V;t)\right)^*\right)\xi\right\Vert=0.
\end{align}
   From (\ref{ttt}), we also have
    \begin{align}\label{rashi}
   \left\Vert
       \left({\mathcal T}_n(V_N;t)
            \sigma_t(V_N)
            \left({\mathcal T}_m(V_N;t)\right)^*-
            {\mathcal T}_n(V;t)
            \sigma_t(V)
            \left({\mathcal T}_m(V;t)\right)^*\right)\xi\right\Vert
            \le \frac{M_\varepsilon^{n+m}}{n!m!}\left(\Vert V\Vert^{n+m+1}+c^{n+m+1}\right)\left\Vert\xi\right\Vert,
   \end{align}
   for any $|t|\le M_\varepsilon$.
      From (\ref{aza}) and (\ref{rashi}), applying the Lebesgue dominated convergence Theorem,
      we have
      \begin{align}
      \lim_{N\to\infty}
      \int_{|t|\le M_\varepsilon}
      dt\ f(t)
       \sum_{n\geq 0} \sum_{m\geq 0}\left| s^{n+m}\right|\left\Vert
       \left({\mathcal T}_n(V_N;t)
            \sigma_t(V_N)
            \left({\mathcal T}_m(V_N;t)\right)^*-
            {\mathcal T}_n(V;t)
            \sigma_t(V)
            \left({\mathcal T}_m(V;t)\right)^*\right)\xi\right\Vert=0.
      \end{align}
  Substituting this to (\ref{fuku}), we obtain
  \begin{align}
\limsup _{N\to\infty}  \left\Vert \left(\Phi(V_N;{s}) - \Phi(V;s)\right)\xi \right\Vert\le
\varepsilon \left\Vert \xi\right\Vert    (\left\Vert V\right\Vert  +c),
  \end{align}
  for any $\varepsilon>0$, $\xi\in \mathcal H$, and $s\in\mathbb R$.
  As $\varepsilon$ is arbitrary, we have
  \begin{align}
  \lim _{N\to\infty}  \left\Vert \left(\Phi(V_N;{s}) - \Phi(V;s)\right)\xi \right\Vert=0,
  \end{align}
  for any $\xi\in \mathcal H$, and $s\in\mathbb R$.
   This proves the third statement.
   \end{proof}

\appendix 

\section{$W^*$-dynamical systems}\label{dyn}
We collect known facts on $W^*$-dynamical systems. See \cite{BR2} and \cite{DJP} and references therein
for more information.
    Let $\mathcal{H}$ be a Hilbert space and $\mathfrak{N} \subset B(\mathcal{H})$ a von Neumann algebra. 
    We denote by $\mathfrak N_*$ the predual of $\mathfrak N$.
We also use $\mathfrak N_{\mathrm sa}$ the set of all selfadjoint elements in $\mathfrak N$.
    Let ${\mathrm{Aut}}(\mathfrak{N})$ denote the group of all $*$-automorphisms of the von Neumann algebra $\mathfrak{N}$. A family $\{\tau_t\}_{t\in \mathbb{R}} \subset {\mathrm{Aut}}(\mathfrak{N})$ is called $W^*$-dynamics if 
    it is satisfied that for all $t$, $s \in \mathbb{R}$, $\tau_s \circ \tau_t = \tau_{s+t}$,   
    $\tau_0 = 1$ and $\tau_t(A)$ converges to $A$ in the $\sigma$-weak topology as $t\to 0$. 
    As $\tau_t$ is an automorphism, for any $A\in\mathfrak N$, ${\mathbb R}\ni t\mapsto \tau_t(A)\in{\mathfrak N}$ is continuous with respect to the $\sigma$-strong topology.
    The pair $(\mathfrak{N}, \tau)$ is called a $W^*$-dynamical system.
    For an entire $\tau$-analytic $A\in \mathfrak{N}$,  we denote by $\tau_z(A)$, for $z\in{\mathbb C}$
    the analytic continuation of ${\mathbb R}\ni t\mapsto \tau_t(A)\in \mathfrak{N}$.
    We denote the set of all $\tau$-entire analytic elements in $\mathfrak{N}$
    by  $\mathfrak{N}_{\tau}$. 
    Then $\mathfrak{N}_{\tau}$ is a $*$-subalgebra in $\mathfrak{N}$,
    which is dense in $\mathfrak{N}$ with respect to 
        the strong topology.
        Let $\beta \in \mathbb{R}$.
        A normal state $\omega$ on $\mathfrak{N}$ is called a ($\tau$,\ $\beta$)-KMS state if 
        for all $A\in \mathfrak{N}$ and $B\in \mathfrak{N}_\tau$, 
        \begin{equation}
            \omega(A\tau_{i\beta}(B)) = \omega(BA).
        \end{equation}
    Any ($\tau$,$\beta$)-KMS state is $\tau$-invariant.
    
    Finally, we consider the perturbation theory of the dynamics. Let $(\mathfrak{N},\tau)$ be a $W^*$-dynamical system and $Q\in \mathfrak{N}_{sa}$. For $t\in \mathbb{R}$, we define $\tau_t^Q$  by 
    \begin{equation}
        \tau_t^Q(A) = \sum_{n\geq 0}(it)^n\int_{0\leq s_n\leq \cdots \leq s_1 \leq 1}
        ds_1 \cdots ds_n\ [\tau_{s_n t}(Q),[\cdots,[\tau_{s_1 t}(Q),\tau_t(A)]\cdots]].
    \end{equation}
    for all $A\in \mathfrak{N}$.
    The integral is defined in the $\sigma$-weak topology and
    the summation is in the norm topology.\\
    Then $\tau^Q$ is a $W^*$-dynamics.
    For all $t\in \mathbb{R}$, define $E_Q^\tau(t) \in \mathfrak{N}$ by 
    \begin{align}\label{eqexp}
        E_Q^\tau(t) &= 
        \sum_{n\geq 0}(it)^n\int_{0\leq s_n\leq \cdots \leq s_1\leq 1}
        ds_1 \cdots ds_n\ \tau_{s_n t}(Q)\cdots \tau_{s_1 t}(Q) \nonumber\\
        &= \sum_{n\geq 0}i^n\int_{0\leq s_n\leq \cdots \leq s_1\leq t}
        ds_1 \cdots ds_n\ \tau_{s_n}(Q)\cdots \tau_{s_1}(Q).
    \end{align}  The integral is defined in the $\sigma$-weak topology and
    the summation is in the norm topology. 
        This  $E_Q^\tau(t)$ is a unitary element of $\mathfrak{N}$ and 
        \begin{equation}\label{ete}
            \tau_t^Q(A) = E_Q^\tau(t)\tau_t(A)E_Q^\tau(t)^*
        \end{equation}
        for all $t\in \mathbb{R}$ and $A\in \mathfrak{N}$.
         Suppose that there exists a self-adjoint operator $L$ on $\mathcal{H}$ satisfying
        $\tau_t(A) = e^{itL}Ae^{-itL}$ for all $t\in \mathbb{R}$ and $A\in \mathfrak{M}$. 
        Then, for all  $t\in \mathbb{R}$, $A\in \mathfrak{M}$ and $Q\in \mathfrak{M}_{sa}$, 
        \begin{equation}\label{tlq}
            \tau_t^Q (A)= e^{it(L+Q)}Ae^{-it(L+Q)}
        \end{equation}
        and 
        \begin{equation}\label{elq}
            E_Q^\tau(t) = e^{it(L+Q)}e^{-itL}.
        \end{equation}
         If $A,Q$ are entire-analytic for $\tau$, then $A$ is
        $\tau^Q$-entire analytic and ${\mathbb R}\ni t\mapsto E^\tau_Q(t)\in {\mathfrak N}$ has an entire analytic continuation.
        We have the expansion
 \begin{equation}
        \tau_z^Q(A) = \sum_{n\geq 0}(iz)^n\int_{0\leq s_n\leq \cdots \leq s_1 \leq 1}
        ds_1 \cdots ds_n\ [\tau_{s_n z}(Q),[\cdots,[\tau_{s_1 z}(Q),\tau_z(A)]\cdots]],
    \end{equation}
    and 
    \begin{equation}\label{eqtz}
        E_Q^\tau(z) = \sum_{n\geq 0}(iz)^n\int_{0\leq s_n\leq \cdots \leq s_1\leq 1}
        ds_1 \cdots ds_n\ \tau_{s_n z}(Q)\cdots \tau_{s_1 z}(Q)
    \end{equation}
    for all $z\in \mathbb{C}$ respectively.
    The integration and the summation converges in  the norm topology. (\cite{DJP})

     \section{Tomita-Takesaki Theory}\label{kms}
    We collect known facts Tomita-Takesaki Theory. See \cite{BR2} and \cite{DJP} and references therein
for more information.
Let $\mathcal{H}$ be a Hilbert space, $\mathfrak{M} \subset B(\mathcal{H})$ a von Neumann algebra, 
    and $\Omega\in \mathcal{H}$ a cyclic and separating vector for $\mathfrak{M}$. 
       Define the operator $S_\Omega^0$ on $\mathcal{H}$ with the domain $\mathfrak{M}\Omega$ by 
            \begin{equation}
                S_\Omega^0 A\Omega = A^*\Omega,\quad A\in \mathfrak{M}.
            \end{equation}
            Then $S_\Omega^0$ is anti-linear and closable. Let $S_\Omega$ be its closure.
       Let $S_\Omega = J_\Omega \Delta^{\frac{1}{2}}_\Omega$ be the polar decomposition of $S_\Omega$. 
          The operators $\Delta_\Omega$ and $J_\Omega$ are called respectively 
    the modular operator and the modular conjugation for $(\mathfrak{M},\Omega)$.
           The modular operator $\Delta_\Omega$ is non-singular.  The modular operator and the modular conjugation satisfy 
                       $J_\Omega \Delta_\Omega J_\Omega  = \Delta_\Omega^{-1}$,
        $\Delta^{z}_{\Omega} \Omega = \Omega$ for all $z\in \mathbb{C}$ and $\log \Delta_\Omega \Omega = 0$.  Furthermore, $J_\Omega$ interchange $\mathfrak{M}$ and $\mathfrak{M}'$, i.e., $J_\Omega \mathfrak{M}J_\Omega  = \mathfrak{M}^{'}$.
        From the fact that 
            $\Delta_\Omega^{it} \mathfrak{M}\Delta_\Omega^{-it} = \mathfrak{M}$
            all $t\in \mathbb{R}$, 
         we can define a $*$-automorphism $\sigma_t$ on $\mathfrak{M}$ by 
            $\sigma_t(A) = \Delta_\Omega^{it} A \Delta_\Omega^{-it}$, $A\in{\mathfrak M}$.
            This $W^*$-dynamics $\sigma$ is called the modular automorphism associated to $\Omega$.
            We set $L_\Omega = \log \Delta_\Omega$
           and call it the Liouvillean of $\sigma$. 
            One can check that
            the normalized state of $\omega$ is a $(\sigma,\ -1)$-KMS state.

        The set $ \mathfrak{M}_\sigma \Omega$ is a core for $S_{\Omega}$ and $\Delta_{\Omega}^{\frac{1}{2}}$.(Proposition 2.5.22 of \cite{BR1}.)
       We use the following lemma.
    \begin{lem}\label{p10}
        Let $Q\in \mathfrak{M}$, $n\in \mathbb{N}$ and $t_1,\ldots,t_n \in \mathbb{R}$ . Let $\{Q_m\}_{m\in \mathbb{N}}\subset \mathfrak{M}$ satisfy that 
        $Q_m \to Q$ and $Q_m^* \to Q^*$ strongly as $m\to\infty$. Then 
        \begin{equation}
            \Delta_\Omega^{it_n}Q_m \cdots \Delta_\Omega^{it_1}Q_m\xi\to 
            \Delta_\Omega^{it_n}Q \cdots \Delta_\Omega^{it_1}Q\xi
        \end{equation}
        as $m\to\infty$, for any $\xi\in\mathcal H$.
    \end{lem}

The following theorem is shown in \cite{A}.
    \begin{thm}\label{p9}
        Let $Q\in \mathfrak{M}_{sa}$. Then $\Omega \in D\left(e^{\frac{L_\Omega+Q}{2}}\right)$.\\
        Define the vectors $\Omega_Q := {e^{\frac{L_\Omega+Q}{2}}\Omega}$, 
        and 
        the functional $\omega_Q$ on $\mathfrak{M}$ by 
        \begin{equation}
            \omega_Q(A) = (\Omega_Q,\ A\Omega_Q) 
        \end{equation}
        for all $A\in \mathfrak{M}$.
        Then the following conditions are satisfied;
        \begin{enumerate}
            \renewcommand{\labelenumi}{(\alph{enumi})}
            \item  $\Omega_Q$ is cyclic and separating for $\mathfrak{M}$;
            \item The normalized state of $\omega_Q$ is a ($\sigma^Q$,\ $-1$)-KMS state;
            \item By (a), we can define the modular operator $\Delta_{\Omega_Q}$ 
            and the modular conjugation $J_{\Omega_Q}$ for the pair $(\mathfrak{M},\ \Omega_Q)$.
                Then $J_\Omega = J_{\Omega_Q}$,
                \begin{equation}
                    \log \Delta_{\Omega_Q} = L_\Omega+Q-J_\Omega QJ_\Omega 
                \end{equation}
                and
                \begin{equation}
                    \sigma_t^Q(A) = \Delta_{\Omega_Q}^{it} A \Delta_{\Omega_Q}^{-it} 
                    = e^{it(L_\Omega+Q)}Ae^{-it(L_\Omega+Q)};
                \end{equation}
            \item Assume that a sequence $\{Q_n\}_{n\in \mathbb{N}}\subset \mathfrak{M}_{sa}$ 
            converges to $Q \in \mathfrak{M}_{\mathrm sa}$ strongly.
            Then $\Omega_{Q_n} \to \Omega_Q$ and $\omega_{Q_n} \to \omega_Q$ in the norm topology;
            \item For any $0\leq s_n\leq \cdots \leq s_1\leq 1$, 
            $\Omega$ is in the domain of $e^{\frac 12 s_nL}Q\cdot e^{\frac 12(-s_n+s_{n-1})L}Q\cdots 
        e^{\frac12 (-s_2+s_1)L}Q$ and 
        $\Omega_Q$ has an expansion        
            \begin{equation}
       \sum_{n\geq 0}\left( \frac{1}{2}\right)^n\int_{0\leq s_n\leq \cdots \leq s_1\leq 1}
        ds_1 \cdots ds_n e^{\frac 12 s_nL}Q\cdot e^{\frac12 (-s_n+s_{n-1})L}Q\cdots 
        e^{\frac12 (-s_2+s_1)L}Q\Omega.
    \end{equation}

        \end{enumerate}
    \end{thm} 
    Due to (c) of this theorem, let us denote $J$ the modular conjugation 
    if there is no danger of confusion.
    
    \section{Properties of $f$}\label{fprop}
    In this section we consider some properties of $f$.
 Clearly $f$ is in $L^1(\mathbb R)$
%    \begin{equation}
%        \int_\mathbb{R}dtf(t) = \sum_{n=0}^\infty \frac{2}{\pi\left(n+\frac{1}{2}\right)^2}   , 
%    \end{equation}
    by Fubini's Theorem.
    This $f$ and $F$ in (\ref{Fdef}) are Fourier transform of each other.
   \begin{lem}\label{ap}
    For  $f$ defined in (\ref{fdef}) and $F$ defined in (\ref{Fdef}), we have
    \begin{equation}\label{Ff}
        \lim_{R\to \infty} \int_{-R}^R F(w)e^{-iwt}dw = f(t)    
    \end{equation}
    for all $t \in \mathbb{R}\setminus\{0\}$ and 
    \begin{equation}\label{ff}
        F(x) = \int_{-\infty}^{\infty} f(t)e^{ixt}dt
    \end{equation}
    for all $x\in \mathbb{R}$.
\end{lem}
\begin{proof}
In order to prove this we consider the analytic continuation of $F$ (defined in (\ref{Fdef})) to
    $\mathbb{C}\setminus \left(2\pi i\left(\mathbb{Z}+\frac{1}{2}\right)\right)$
    :
     \begin{equation}
        F(z) = \begin{cases}\displaystyle
            \frac{e^z-1}{e^z+1} \frac{1}{z} & \left(z \in\mathbb{C}\setminus \left(2\pi i\left(\mathbb{Z}+\frac{1}{2}\right)\cup\{0\}\right)\right)   \\
            \\
            \displaystyle\frac{1}{2} & (z = 0).
        \end{cases}
    \end{equation}
The claim (\ref{Ff}) for $t<0$ follows applying the residue theorem 
to $F(z)e^{-izt}$ along the path in the following picture.
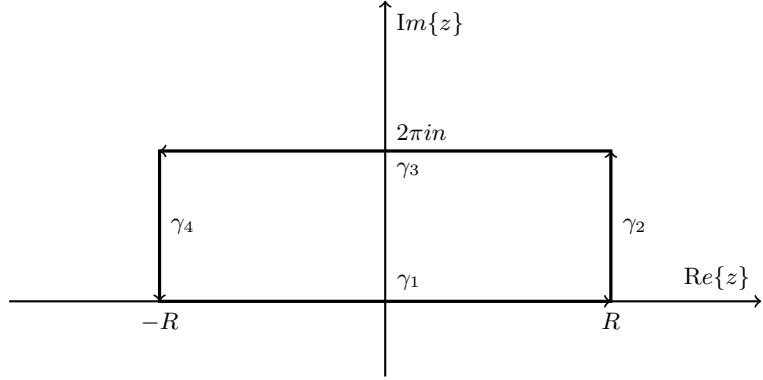
\begin{center}
    \begin{tikzpicture}
        \begin{scope}[thick,font=\scriptsize]
        \draw [->] (-5,0) -- (5,0) node [above left]  {$\mathrm Re\{z\}$};
        \draw [->] (0,-1) -- (0,4) node [below right] {$\mathrm Im\{z\}$};
    
        \iffalse% Single
        \else
        \draw(-3,0)[very thick] rectangle (3,2);
        \draw [->] (-3,0)[very thick] -- (0,0) node [above right]  {$\gamma_1$} [thick]-- (3,0) ;
        \draw [->] (3,0)[very thick] -- (3,1) node [right]  {$\gamma_2$} [thick]-- (3,2) ;
        \draw [->] (3,2)[thick] -- (0,2) node [below right]  {$\gamma_3$} [thick]-- (-3,2) ;
        \draw [->] (-3,2)[thick] -- (-3,1) node [right]  {$\gamma_4$} [thick]-- (-3,0) ;
    
        \draw(-3,0) node [below] {$-R$};
        \draw(3,0) node [below] {$R$};
        \draw(0,2) node[above right] { $2\pi in$};
        \fi
        \end{scope}
    \end{tikzpicture}\\
    \end{center}
Each pole $z = 2\pi i\left(n+\frac{1}{2}\right)$ of $F(z)e^{-itz}$ is a simple pole of $F(z)e^{-izt}$, and
    \begin{align}
        \mathrm{Res}\left(F(z)e^{-izt};2\pi i\left(n+\frac{1}{2}\right)\right)
        &= \lim_{z\to 2\pi i\left(n+\frac{1}{2}\right)}\left(z-2\pi i\left(n+\frac{1}{2}\right)\right)F(z)e^{-izt}\nonumber\\
        &= -\frac{i}{\pi \left(n+\frac{1}{2}\right)}e^{2\pi \left(n+\frac{1}{2}\right)t}.
    \end{align}
 From the residue Theorem, we obtain
     \begin{align}
        \int_{-R}^{R} dw\ F(w)e^{-iwt}
        &= \int_{\gamma_1}F(z)e^{-izt}dz\nonumber\\
        &= 2\pi i\sum_{k=0}^{n-1} \mathrm{Res}\left(F(z)e^{-izt};2\pi i\left(k+\frac{1}{2}\right)\right)\nonumber\\
        & \qquad -\int_{\gamma_2}F(z)e^{-izt}dz-\int_{\gamma_3}F(z)e^{-izt}dz-\int_{\gamma_4}F(z)e^{-izt}dz.
    \end{align}
    The integral along $\gamma_3$ goes to $0$ as $n\to\infty$.
The integral along $\gamma_2,\gamma_4$ goes to $0$ as $R\to\infty$.
Hence we obtain (\ref{Ff}) for all $t<0$.

To show that (\ref{Ff}) for $t>0$, we note
    \begin{align}
        \int^R_{-R}F(w)e^{-iwt}dw 
        = \overline{\int^R_{-R}\overline{F(w)}e^{-iw(-t)}dw} = \overline{\int_{-R}^R F(w)e^{-iw(-t)}dw} .
    \end{align}
    The final term converges to $f(-t) = f(t)$ as $R\to \infty$ due to $-t<0$.
    Hence we obtain (\ref{Ff}) for all $t\neq 0$.
    
    Finally we show (\ref{ff}).
%    \begin{equation}
%        F(x) = \int_{-\infty}^\infty f(t)e^{ixt}dt.    
%    \end{equation}
%    Because we have
%    \begin{align}
%        \int_{\mathbb{R}}dx |F(x)|^2 &=
%        \int_{|x|\geq 1}dx\left|\frac{e^x-1}{e^x+1}\right|^2\frac{1}{x^2} + \int_{|x|\leq 1}dx|F(x)|^2 \nonumber\\
%        &\leq \int_{|x|\geq 1}dx\frac{1}{x^2} + \int_{|x|\leq 1}dx|F(x)|^2 < \infty,
%    \end{align}
   Note that $F|_{\mathbb{R}} \in L^2(\mathbb{R})$. 
    Let $\hat{F} \in L^2(\mathbb{R})$ be the Fourier transform of $F|_{\mathbb{R}}$.
    For $R > 0$, we define the function $\phi_R$ on $\mathbb{R}$ by 
    \begin{equation}
        \phi_R(t) = \int_{-R}^R F(x)e^{-ixt}dx.
    \end{equation}
    for all $t\in \mathbb{R}$.
    Then we have 
    $|| \phi_R - \hat{F}||_2 \to 0$ as $R\to \infty$ by Theorem 9.13 of \cite{R}.
    Therefore there exists a sequence $\{n_k\}_{k\in\mathbb{N}}$ such that $\phi_{n_k}(t) \to \hat{F}(t)$ as $k\to\infty$
    a.e. $t$.
    Since we have shown \begin{equation}\lim_{R\to \infty} \phi_R(t)  = f(t)\end{equation} for all $t\in \mathbb{R}\setminus \{0\}$, 
    it follows that $f(t) = \hat{F}(t)$ a.e. $t$.
    Hence $\hat{F} = f$ belongs to $L^1(\mathbb{R})$.
    By Theorem 9.14 of \cite{R}, we have
    \begin{equation}
    F(x) = \int_{-\infty}^{\infty} \hat{F}(t)e^{ixt} dt = \int_{-\infty}^{\infty} f(t)e^{ixt}dt   
    \end{equation} a.e. $x$.
    The left-hand side is  continuous. The right-hand side is continuous too due to $f\in L^1(\mathbb{R})$.
    Therefore the equality holds for all $x\in\mathbb R$.

\end{proof}

\end{document}